\documentclass[10pt,journal,compsoc,twoside]{IEEEtran}
\newtheorem{theorem}{Theorem}

\newtheorem{proposition}[theorem]{Proposition}
\newtheorem{problem}[theorem]{Problem}
\newtheorem{remark}[theorem]{Remark}
\newtheorem{definition}[theorem]{Definition}

\newtheorem{proof}{Proof}

\usepackage{dsfont}
\usepackage{color,amsmath,amssymb}
\newcommand{\fM}{{\mathfrak M}}
\newcommand{\support}{{\rm Supp}}
\newcommand{\diag}{{\rm Diag}}

\newcommand{\D}{{\mathfrak D}}

\ifCLASSOPTIONcompsoc
  \usepackage[nocompress]{cite}
\else
  \usepackage{cite}
\fi
\ifCLASSINFOpdf
\else
\fi
\begin{document}
%
\title{Optimal steering to invariant distributions\\ for networks flows}
%
%
%
%

\author{Yongxin~Chen,~\IEEEmembership{Member,~IEEE,} 
        Tryphon~T. ~Georgiou,~\IEEEmembership{Fellow,~IEEE,}
        and~Michele~Pavon
\IEEEcompsocitemizethanks{\IEEEcompsocthanksitem Y. Chen is with the School of Aerospace Engineering, Georgia Institute of Technology, Atlanta, GA 30332, USA, T.T. Georgiou is with the Department of Mechanical and Aerospace Engineering,University of California, Irvine, CA 92697, USA, M. Pavon is with the Department of Mathematics ``Tullio Levi-Civita",Universit\`a di Padova, 35121 Padova, Italy.}}

%
%

\markboth{Optimal steering to invariant distributions for networks flows} 
{Chen, Georgiou, and Pavon} 
%



\IEEEtitleabstractindextext{%
\begin{abstract}
We derive novel results on the ergodic theory of irreducible, aperiodic Markov chains. We show how to optimally steer the network flow to a stationary distribution over a finite or infinite time horizon. Optimality is with respect to an entropic distance between distributions on feasible paths. When the prior is reversible, it shown that solutions to this discrete time and space steering problem are reversible as well. A notion of temperature is defined for Boltzmann distributions on networks, and problems analogous to cooling (in this case, for evolutions in discrete space and time) are discussed. 
\end{abstract}

}

\maketitle

\IEEEdisplaynontitleabstractindextext

%
\IEEEpeerreviewmaketitle

\IEEEraisesectionheading{\section{Introduction}\label{sec:introduction}}

%
%
%
%
We consider an optimal steering problem for networks flows over a finite or infinite time horizon. Specifically, we derive discrete counterparts of our {\em cooling} results in \cite{CGPcooling}. The goal is to steer the Markovian evolution to a steady state with desirable properties while minimizing relative entropy, or relative entropy rate, over admissible distributions on paths. This relates to a special Markov Decision Process problem, cf. \cite[Section 6]{CGPAnnualReview} which is referred to as {\em Schr\"{o}dinger's Bridge Problem} (SBP) or {\em regularized optimal mass transport} (OMT). A rich history on this circle of ideas originates from two remarkable papers by Erwin Schr\"odinger in 1931/32 \cite{S1,S2} who was interested in a large deviation problem for a cloud of Brownian particles, and of earlier work by Gaspar Monge in 1781 on the problem of optimally transporting mass between sites \cite{Monge}.

Important contributions on the existence question in SBP, left open by Schr\"odinger, were provided over time by Fortet, Beurling, Jamison and F\"ollmer \cite{For,Beu,Jam2,F2}.
It should be remarked that Fortet's proof in 1940 is algorithmic, establishing convergence of a (rather complex) iterative scheme. This predates by more than twenty years the contribution of Sinkhorn \cite{Sin64} who established convergence in a special case of the discrete problem The latter proof is much simpler than in the continuous (time and space) problem solved by Fortet. Thus, these algorithms should be called Fortet-IPF-Sinkhorn, where IPF stands for the Iterative Proportional Fitting algorithm proposed in 1940 without proof of convergence in \cite{DS1940}. These problems admit a ``static" and a ``dynamic" formulation. While the static, discrete space problem was studied by many, starting from Sinkhorn, see \cite{Cuturi,PC} and references therein, the dynamic, discrete problem was considered in \cite{PT,GP,CGPT1,CGPT2,CGPT3}. We provide in this paper a new fluid-dynamic derivation of the Schr\"odinger system on which the iteration is based.
Convergence of the iterative scheme in a suitable projective metric was also established in \cite{GP} for the discrete case and in \cite{CGP_SIAMAppl.} for the continuous case.

This topic lies nowadays at the crossroads of many fields of science such as probability, statistical physics, optimal mass transport, machine learning, computer graphics, statistics, stochastic control, image processing, etc. Several survey papers have  appeared over time emphasizing different aspects of the subject, see \cite{W,leo2,PC,CGP_SIREV,CGPAnnualReview,CGP_CSM}. 

The paper is outlined as follows. 
In Section \ref{GBProblems}, we present background on the discrete-time and space Schr\"odinger bridge problem while emphasizing its ``fluid dynamic'' formulation (the original proof of a key well-known result is deferred to the Appendix).
In Section \ref{OSSS}, we introduce the corresponding 
infinite-horizon steering problem by minimizing the entropy rate with respect to the prior measure. We then discuss existence for the one-step Schr\'odinger system leading to existence for the steering problem. Section \ref{REVERSE} is devoted to an interesting result linking optimality in the steering problem to reversibility of the solution. The final Section \ref{cooling} applies all the previous results to cooling, where the goal is to optimally steer the Markov chain to a steady state corresponding to a lower {\em effective} temperature.

\section{Discrete-time dynamic bridges: background}\label{GBProblems}

We begin by discussing a paradigm of great significance in network flows. It amounts to designing probabilistic transitions between nodes, and thereby probability laws on path spaces, so as to reconcile marginal distributions with priors that reflect on the structure of the network and objectives on transferance of resourses across the network. The basic formulation we discuss amounts to the so called Schr\"odinger bridge problem in discrete-time. This dynamic formulation echoes the fluid dynamic formulation of the classical (continuous time and space) Schr\"odinger bridge problem. Certain proofs that we present are mildly new.

We consider a directed, strongly connected, aperiodic graph ${\bf G}=(\mathcal X,\mathcal E)$ with vertex set $\mathcal X=\{1,2,\ldots,n\}$ and edge set $\mathcal E\subseteq \mathcal X\times\mathcal X$, 
and we consider trajectories/paths on this graph over the time set ${\mathcal T}=\{0,1,\ldots,N\}$. The family of feasible paths $x=(x_0,\ldots,x_N)$ of length $N$, namely paths such that $x_ix_{i+1}\in\mathcal E$ for $i=0,1,\ldots,N-1$,   
is denoted by ${\mathcal FP}_0^N\subseteq\mathcal X^{N+1}$.
We seek a probability distribution $\mathfrak P$ on the space of paths ${\mathcal {FP}}_0^N$ with prescribed initial and final marginals
$\nu_0(\cdot)$ and $\nu_N(\cdot)$, respectively, and such that the resulting random evolution
is closest to a ``prior'' measure $\fM$ on ${\mathcal {FP}}_0^N$ in a suitable sense.

The prior law for our problem is a Markovian evolution
 \begin{equation}\label{FP}
\mu_{t+1}(x_{t+1})=\sum_{x_t\in\mathcal X} \mu_t(x_t) m_{x_{t}x_{t+1}}(t)
\end{equation}
with nonnegative distributions $\mu_t(\cdot)$ over $\mathcal X$, $t\in{\mathcal T}$, and weights $m_{ij}(t)\geq 0$ for all indices $i,j\in{\mathcal X}$ and all times.
In accordance with 
the topology of the graph, $m_{ij}(t)=0$ for all $t$ whenever $ij\not\in\mathcal E$.  Often, but not always, the matrix
\begin{equation}\label{eq:matrixM}
M(t)=\left[ m_{ij}(t)\right]_{i,j=1}^n
\end{equation}
does not depend on $t$. 
Here, $m_{ij}$ plays the same role as that of the heat kernel in continuous-time.
However, herein, the transition matrix $M(t)$ may not represent a {\em probability} transition matrix
in that the rows of $M(t)$ do not necessarily sum up to one. Thus, the ``total transported mass'' is not necessarily preserved,
in that $\sum_{x_t\in\mathcal X} \mu_t(x_t)$ may be afunction of time.
 A particular case of interest is when $M$ is the {\em adjacency matrix} of the graph that encodes the topological structure of the network.

The evolution \eqref{FP}, together with measure $\mu_0(\cdot)$, which we assume positive on $\mathcal X$, i.e.,
\begin{equation}\label{eq:mupositive}
\mu_0(x)>0\mbox{ for all }x\in\mathcal X,
\end{equation}
 induces
a measure $\fM$ on ${\mathcal {FP}}_0^N$ as follows. It assigns to a path  $x=(x_0,x_1,\ldots,x_N)\in{\mathcal {FP}}_0^N$ the value
\begin{equation}\label{prior}\fM(x_0,x_1,\ldots,x_N)=\mu_0(x_0)m_{x_0x_1}(0)\cdots m_{x_{N-1}x_N}(N-1),
\end{equation}
and gives rise to a flow
of {\em one-time marginals}
\[\mu_t(x_t) = \sum_{x_{\ell\neq t}}\fM(x_0,x_1,\ldots,x_N), \quad t\in\mathcal T.\]
\begin{definition} We denote by ${\mathcal P}(\nu_0,\nu_N)$ the family of probability distributions on ${\mathcal {FP}}_0^N$ having the prescribed marginals $\nu_0(\cdot)$ and $\nu_N(\cdot)$.
\end{definition}

We seek a distribution in this set which is closest to the prior $\fM$ in {\em relative entropy} (divergence, Kullback-Leibler index)  defined by 
\begin{equation*}
\D(\mathfrak P\|Q):=\left\{\begin{array}{ll} \sum_{x}\mathfrak P(x)\log\frac{\mathfrak P(x)}{\mathfrak M(x)}, & \support (\mathfrak P)\subseteq \support (\mathfrak M),\\
+\infty , & \support (\mathfrak P)\not\subseteq \support (\mathfrak M),\end{array}\right.
\end{equation*}
Here, by definition,  $0\cdot\log 0=0$.
The value of $\D(\mathfrak P\|\mathfrak M)$ may turn out to be negative due to the different total masses in the case when $\fM$ is not a probability measure. The optimization problem, however, poses no challenge  as the relative entropy  is  (jointly) convex over this larger domain and is bounded below.
This brings us to the 
so-called
{\em Schr\"odinger Bridge  Problem} (SBP):

\begin{problem}\label{prob:optimization}
Determine
 \begin{eqnarray}\label{eq:optimization}
\mathfrak P^*[\nu_0,\nu_N]:={\rm argmin}\{ \D(\mathfrak P\|\fM) \mid  \mathfrak P\in {\mathcal P}(\nu_0,\nu_N)
\}.
\end{eqnarray}
\end{problem}

By analyzing the contribution to the entropy functional of transitions along edges, in a sequential manner, we arive at a  formulation that parallels the fluid-dynamic formulation of the Schr\"odinger problem in continuous time\footnote{
The fluid dynamic reformulation of  Schr\"odinger's bridge problem in continuous space/time drew important connections with Monge-Kantorovich optimal mass transport (OMT) \cite{BB} --the latter seen as a limit when the stochastic excitation reduces to zero  \cite{leo2}. In particular, it pointed to
a time-symmetric formulation of Schr\"odinger's bridges \cite{CGPJOTA}.}, considered in \cite[Section 4]{leo2} and \cite{CGPJOTA}.

Suppose we restrict our search in Problem \ref{prob:optimization} to Markovian $\mathfrak P$'s. In analogy to (\ref{prob:optimization}), we then have
\begin{equation}\label{markovianmeasure}
\mathfrak P(x_0,x_1,\ldots,x_N)=\nu_0(x_0)\pi_{x_0x_1}(0)\cdots \pi_{x_{N-1}x_N}.
\end{equation}
Also let $p_t$ be the one-time marginals of $\mathfrak P$. i.e.
\[p_t(x_t) = \sum_{x_{\ell\neq t}}\mathfrak P(x_0,x_1,\ldots,x_N), \quad t\in\mathcal T.\] 
We finally have the update mechanism
\begin{equation}\label{update}
p_{t+1}(x_{t+1})=\sum_{x_t\in\mathcal X} p_t(x_t) \pi_{x_{t}x_{t+1}}(t)
\end{equation}
which, in vector form, is 
\begin{equation}
\label{eq:fp}
p_{t+1}=\Pi'(t) p_{t}.
\end{equation}
Here $\Pi=\left[ p_{ij}(t)\right]_{i,j=1}^n$ is the transition matrix and prime denotes transposition.
Using (\ref{prior})-(\ref{markovianmeasure}) we obtain
\[\D(\mathfrak P\|\fM)=\D(\nu_0\|\mu_0) + \sum_{t=0}^{N-1}\sum_{x_t}\D(\pi_{x_tx_{t+1}}(t)\|m_{x_tx_{t+1}}(t))p_t(x_t).
\]
Since $\D(\nu_0\|\mu_0)$ is invariant over ${\mathcal P}(\nu_0,\nu_N)$, we can now have the following equivalent (fluid dynamic) formulation:
\vspace{.3cm}
\begin{problem}\label{problemfluid-dynamic}
 \begin{align}
{\rm min}_{(p,\pi)}&\left\{\sum_{t=0}^{N-1}\sum_{x_t}\D\left(\pi_{x_tx_{t+1}}(t)\|m_{x_tx_{t+1}}(t))p_t(x_t)\right)\right\},\nonumber\\\label{fluid-dynamic2}{\rm \hspace*{-15pt}subject\; to}\; & \nonumber\\
&p_{t+1}(x_{t+1})=\sum_{x_t\in\mathcal X} p_t(x_t) \pi_{x_{t}x_{t+1}}(t),\\
&  \sum_{x_{t+1}}\pi_{x_{t}x_{t+1}}(t)=1, \; \forall x_t\in\mathcal X, \mbox{ and}\label{fluid-dynamic3}\\
&p_0(x_0)=\nu_0(x_0), \quad p_N(x_N)=\nu_N(x_N),\label{fluid-dynamic4}\\
&\mbox{for } t\in\{0,1,\ldots,N-1\}.
\end{align}
\end{problem}
\noindent
Here $p=\{p_t; t=0,1,\ldots,N\}$ is a flow of probability distributions on $\mathcal X$ (corresponds to the flow of densities $\rho_t$ in the continuous setting) and $\pi=\{\pi_{x_tx_{t+1}}(t); t=0,1,\ldots,N-1\}$ is a flow of transition probabilities (corresponds to the flow of the drifts). 
The relative entropy $\D\left(\pi_{x_tx_{t+1}}(t)\|m_{x_tx_{t+1}}(t))\right)$ between transition laws, corresponds to the
square of the modulus of the velocity field in the continuous setting \cite{CGPJOTA}, so that $\D\left(\pi_{x_tx_{t+1}}(t)\|m_{x_tx_{t+1}}(t))p_t(x_t)\right)$ in fact corresponds to kinetic energy.

We do not impose explicitly the non-negativity constraint on the $p_t$ and $\pi_{x_tx_{t+1}}(t)$ as
it will be automatically satisfied  (due to the fact that the relative entropy grows unbounded at the boundary of the admissible distributions and transition probabilities). We have the following theorem:

\begin{theorem} \label{maincl}Suppose there exists a pair of nonnegative functions $(\varphi,\hat{\varphi})$ defined  on $[0,T]\times{\cal X}$ and satisfying the system
\begin{eqnarray}\label{scsi1}\varphi(t,x_t)=\sum_{x_{t+1}}m_{x_{t}x_{t+1}}(t)\varphi(t+1,x_{t+1}),\\
\hat{\varphi}(t+1,x_{t+1})=\sum_{x_t}  m_{x_{t}x_{t+1}}(t)\hat{\varphi}(t,i),\label{scsi2}
\end{eqnarray}
for $0\le t \le(T-1)$, as well as the boundary conditions
\begin{equation}\label{scsi3}
\varphi(0,x_0)\cdot\hat{\varphi}(0,x_0)=\nu_)(x_0),\; \varphi(T,x_T)\cdot\hat{\varphi}(T,x_T)=\nu_T(x_T),
\end{equation}
for $x_t\in\mathcal X$ and $t\in\{0,T\}$, accordingly.
Suppose moreover that $\varphi(t,i)>0,\; \forall 0\le t\le T, \forall i\in{\cal X}$. Then, the Markov distribution $\mathfrak P^*$ in ${\mathcal P}(\nu_0,\nu_N)$ having transition probabilities
\begin{equation}\pi^*_{x_{t}x_{t+1}}(t)=m_{x_{t}x_{t+1}}(t)\frac{\varphi(t+1,x_{t+1})}{\varphi(t,x_{t})}
\end{equation}
solves Problem \ref{problemfluid-dynamic}. 
\end{theorem}
\begin{remark}\label{uniqueness}Notice that if $(\varphi,\hat{\varphi})$ satisfy (\ref{scsi1})-(\ref{scsi2})-(\ref{scsi3}), so does the pair $(c\varphi,\frac{1}{c}\hat{\varphi})$ for all $c>0$. Hence, uniqueness for the Schr\"{o}dinger system is always intended as uniqueness of the ray.
\end{remark}
The proof of Theorem \ref{maincl} is deferred to Appendix \ref{A}.
Our derivation there of the Schr\"odinger system appears original and alternative to the two usual approaches. The first one uses Lagrange multipliers for the constraints involving the marginals \cite{GP} very much as in Schr\"odinger's original spirit. The second uses a ``completion of the relative entropy" argument, see \cite[Section IV]{PT}. Let $M(t)=\left(m_{ij}(t)\right)$. Under the assumption that  the entries of the matrix product 
\begin{equation}\label{posassumption}G:=M(0)M(1) \cdots M(N-2)M(N-1)
\end{equation} 
are all positive, there exists a (unique in the projective geometry sense) solution to the system (\ref{scsi1})-(\ref{scsi3}) which can be computed through a Fortet-IPF-Sinkhorn iteration \cite{For,DS1940,Sin64,Cuturi,GP}.

\section{Optimal Steering to a Steady State}\label{OSSS}

Consider a given prior transition rate between the vertices of the graph, represented by a matrix $M$ as in (\ref{eq:matrixM}), which at present we assume as being time-invariant. Suppose $\pi$ is a desired stationary  probability distribution that we wish to maintain over the state space $\mathcal X$ for the distribution of resources, goods, packets, vehicles, etc.\ at steady state. We address below the following basic question: {\em How can we modify  $M$ so that the new evolution is close to the prior in relative entropy while at the same time it admits $\pi$ as invariant distribution?} 
We address the above question next and provide a partial answer
to the complementary question on what conditions ensure that any given $\pi$ can be rendered the stationary distribution with a transition law compatible with the topology of the graph.

\subsection{Minimizing the relative entropy rate}
We consider the following problem inspired by \cite{DL}. Given a (prior) measure $\fM\in {\mathcal {FP}}_0^N$ as in Section \ref{GBProblems}, corresponding to a constant transition matrix $M$, and a probability vector $\pi$ on $\mathcal X$, 
we are seeking a 
(row)
stochastic matrix $\Pi$ such that 
\[
\Pi'\pi=\pi
\]
while the corresponding (time-invariant) measure $\mathfrak P$ on $\mathcal X\times\mathcal X\times\cdots$ is close to $\mathfrak M$ in a suitable sense. Our interest here focuses on the infinite horizon case where $N$ is arbitrarily large.

To this, we let ${\mathcal P}$ be the family of probability distributions on ${\mathcal {FP}}_0^N$, and consider the following problem.
\begin{problem}\label{P1}
\begin{eqnarray}\nonumber
&&\min_{\mathfrak P\in\mathcal P}\lim_{N\rightarrow\infty}\frac{1}{N}\D\left(\mathfrak P_{[0,N]}\|\mathfrak M_{[0,N]}\right)\\
{\rm subject\, to}\; &&\Pi'\pi=\pi,\nonumber\\\nonumber&&\Pi\mathds{1}=\mathds{1},
\end{eqnarray}
where $\mathds{1}$ is the vector with all entries equal to $1$.
\end{problem}
By formula $(18)$ in \cite{PT}, which applies to general (not necessarily mass preserving) Markovian evolutions, we have the following representation for the relative entropy between Markovian measures
\begin{eqnarray}\nonumber
&&\D\left(\mathfrak P_{[0,N]}\|\mathfrak M_{[0,N]}\right)\\\nonumber&&=\D(\pi\|m_0)+\sum_{k=0}^{N-1}\sum_{i_k}\D(\pi_{i_ki_{k+1}}\|m_{i_ki_{k+1}})\pi(i_k).
\end{eqnarray}
 Thus, Problem \ref{P1} reduces to the {\em stationary} Schr\"{o}dinger bridge problem:
 \begin{problem}\label{P2}
 \begin{eqnarray}\nonumber
&& \min_{\left(\pi_{ij}\right)}\sum_{i,j\in\mathcal X}\D(\pi_{ij}\|m_{ij})\pi(i),\\
{\rm subject\, to}\; &&\sum_{i\in\mathcal X} \pi_{ij}\pi(i)=\pi(j),\quad j\in\mathcal X, \nonumber\\
\nonumber&&\sum_{j\in\mathcal X}\pi_{ij}=1, \quad i\in\mathcal X.
\end{eqnarray}
 \end{problem}
 
 Note that in the above we have ignored enforcing the nonnegativity of the $\pi_{ij}$. This  problem can readily be shown to be equivalent to a standard {\em one-step} Schr\"odinger bridge problem  for the joint distributions $p(i,j)$ and $m(i,j)$ at times $t\in\{0,1\}$ with the two marginals equal to $\pi$. Indeed, a straightforward calculation gives
 \begin{eqnarray}\nonumber&&\sum_{ij}\log\frac{\pi(i,j)}{m(i,j)}{\pi(i,j)}=\sum_{ij}\log\frac{\pi_{ij}\pi(i)}{m_{ij}m_0(i)}{\pi_{ij}\pi(i)}\\&&=\sum_{i}\D(\pi_{ij}\|m_{ij})\pi(i)+ \D(\pi\|m_0).\label{SBjoint}
 \end{eqnarray}
 
Following Theorem \ref{maincl}, suppose there exist two vectors
 \[(\varphi(t,\cdot),\hat{\varphi}(t,\cdot)), \mbox{ for } t\in\{0,1\},
 \] 
  with nonnegative entries, satisfying the Schr\"odinger system 
 
 \begin{subequations}\label{eq:DSchroesystem}
\begin{eqnarray}\label{eq:DSchroesystemA}
&&\varphi(0,i)=\sum_j m_{ij}\varphi(1,j),\\
&&\hat{\varphi}(1,j)=\sum_i m_{ij}\hat{\varphi}(0,i)\label{eq:DSchroesystemB},\\&&\varphi(0,i)\cdot\hat{\varphi}(0,i)=\pi(i),\label{eq:DSchroesystemC}\\&&\varphi(1,j)\cdot\hat{\varphi}(1,j)=\pi(j).\label{eq:DSchroesystemD}
\end{eqnarray}
\end{subequations}
Thence, we conclude that
\begin{equation}\label{one-step_opt}\pi^*_{ij}=\frac{\varphi(1,j)}{\varphi(0,i)}m_{ij}
\end{equation} 
satisfies both constraints of Problem \ref{P2}. The corresponding measure $\mathfrak P^*\in\mathcal P$ solves Problem \ref{P1}. There is, however, a {\em difficulty}. The assumption on the matrix $G$ in (\ref{posassumption}), which guarantees existence for the Schr\"odinger system, becomes here the fact that $M$ must have all positive elements. The latter property is typically not satisfied 
since
  $M$ must comply with the topology of the graph. There exists, fortunately, a reasonable condition on $M$ which ensures existence of solutions for (\ref{eq:DSchroesystem}).
 
 \subsection{Existence for the one-step Schr\"odinger system}\label{Existence}
 \begin{definition} \cite{FL} A square matrix $A =(a_{ij})$ is called  {\em indecomposable} if no permutation matrix $P$ exists such that
 \[A=P\left[\begin{array}{cc}A_{11} &0 \\ A_{21} & A_{22}\end{array}\right]P'
 \]
 where $A_{11}$ and $A_{22}$ are nonvacuous square matrices. $A$ is called {\em fully indecomposable} if there exist no pair of permutation matrices $P$ and $Q$ such that 
 \begin{equation}\label{structure}A=P\left[\begin{array}{cc}A_{11} &0 \\ A_{21} & A_{22}\end{array}\right]Q
\end{equation}
where $A_{11}$ and $A_{22}$ are nonvacuous square matrices.
 \end{definition}
 
 \begin{remark} A square, indecomposable matrix $A$ has a real positive simple eigenvalue equal to its {\em spectral radius} \cite{Varga}. Let $A=(a_{ij})$ be the {\em adjacency matrix} of the graph ${\bf G}=(\mathcal X,\mathcal E)$, namely
 \[a_{ij}=\left\{\begin{array}{ll} 1, &(i, j)\in\mathcal E,\\0, & (i,j)\not\in \mathcal E.\end{array}\right.
 \]
Then $A$ is indecomposable if and only if ${\bf G}=(\mathcal X,\mathcal E)$ is strongly connected \cite[p.608]{FL}.
 \end{remark}
 


 \begin{proposition}\label{PROP} Suppose $M$ is fully indecomposable and $\pi$ has all positive components. Then there exists a solution to (\ref{eq:DSchroesystem}) with $\varphi(0,\cdot)$ and $\varphi(1,\cdot)$ with positive components which is unique in the sense of Remark \ref{uniqueness}.
 \end{proposition}
 \begin{proof} Let $\varphi_0=\varphi(0,\cdot)$ and $\varphi_1=\varphi(1,\cdot)$. Observe that (\ref{eq:DSchroesystemA})-(\ref{eq:DSchroesystemB})-(\ref{one-step_opt}) admit the matricial form
 \begin{subequations}\label{eq:DSchroesystemMatrix}
\begin{eqnarray}\label{eq:DSchroesystemMatrixA}\varphi_0&=&M\varphi_1,\\\label{eq:DSchroesystemMatrixB}\hat{\varphi}_1&=&M'\hat{\varphi}_0,\\
\label{eq:DSchroesystemMatrixC}\Pi^*&=&\diag(\varphi_0)^{-1}M\diag(\varphi_1).
\end{eqnarray}
\end{subequations}
 The proof can now be constructed along the lines of  \cite[Theorem 5]{MO}.
 \end{proof}
Another interesting question is the following: Given the graph ${\bf G}$, what distributions $\pi$ admit at least one stochastic matrix compatible with the topology of the graph for which they are invariant? Clearly, if all self loops are present ($(i,i)\in\mathcal E,\forall i\in\mathcal X$), any distribution is invariant with respect to the identity matrix which is compatible. Without such a strong assumption, a partial answer is provided by the following result.


 \begin{proposition}\label{PROP2}Let $\pi$ be a probability distribution supported on all of $\mathcal X$,  i.e. $\pi(i)>0, \forall i\in\mathcal X$. Assume that the adjacency matrix of ${\bf G}=(\mathcal X,\mathcal E)$ is fully indecomposable. Then, there exist stochastic matrices $\Pi$ compatible with the topology of the graph ${\bf G}$  such that
 \[\Pi'\pi=\pi.
 \]
 \end{proposition}
 \begin{proof} Take, in Problem \ref{P2}, $M=A=[a_{ij}]$ the adjacency matrix of the graph ${\bf G}$.
 By Theorem \ref{maincl} and Proposition \ref{PROP}, Problem \ref{P2} has a solution which admits $\pi$ as invariant distribution.
 \end{proof}
 One final interesting question is the following: Assume that the adjacency matrix of the graph ${\bf G}$ is fully indecomposable. Let $\mathcal Q_\pi$ be the non empty set of stochastic matrices $Q$ compatible with the topology of ${\bf G}$ and such that $Q'\pi=\pi$. Clearly $\mathcal Q$ is a convex set. Can we characterize the matrix in $\mathcal Q$ which induces a maximum entropy rate measure on the feasible paths? We already know the answer to this question. Indeed, consider the formulation of Problem \ref{P2} with the distribution of the edges of formula (\ref{SBjoint}). Recall that maximizing entropy (rate) is equivalent to minimizing relative entropy (rate) from the uniform. If we take as prior the uniform distribution on the edges, $M$ is just a positive scalar times the adjacency matrix $A$ of the graph. Thus, the solution to Problem \ref{P2} with $A$ as prior transition provides the maximum entropy rate. 
 
  
 \section{Reversibility}\label{REVERSE}
 In \cite[Corollary 2]{CGPcooling}, we have shown that, in the case of a reversible prior (Boltzmann density) for a stochastic oscillator, the solution of the continuous countepart of Problem \ref{P1} (minimizing the expected input power) is reversible. We prove next that the same remarkable property holds here in the discrete setting.
 \begin{theorem}\label{revers} Let the transition matrix $M$ of the prior measure $\mathfrak M$ be time invariant.  Assume that  $M$ is reversible with respect to $\mu_0=\mu$, i.e.
 \begin{equation}\label{revprior}
 \diag(\mu)M=M'\diag(\mu).
 \end{equation}
 Then, then the solution $\Pi^*$  of Problem  \ref{P2}  is also reversible with respect to $\pi$.
 \end{theorem}
 \begin{proof}
Reversibility of $M$ with respect to $\mu$ is equivalent to the statement that $\Sigma:=\diag(\mu)M$ is symmetric, as noted. Define the vectors
 \[\hat{\psi}_0=\diag(\mu)^{-1}\hat{\varphi}_0, \quad \psi_0=\diag(\mu)\varphi_0.
 \]
Then, the first two equations in system (\ref{eq:DSchroesystemMatrix}) supplemented with the marginal conditions read
  \begin{subequations}\label{eq:DSchroesystemMatrix2}
\begin{eqnarray}\label{eq:DSchroesystemMatrix2A}&&\psi_0=\Sigma\varphi_1,\\\label{eq:DSchroesystemMatrix2B}&&\hat{\varphi}_1=\Sigma'\hat{\psi}_0=\Sigma\hat{\psi}_0,\\
\label{eq:DSchroesystemMatrix2C}&&\psi_0\circ \hat{\psi}_0=\varphi_1\circ \hat{\varphi}_1=\pi,
\end{eqnarray}
\end{subequations}
where the operation $\circ$ in (\ref{eq:DSchroesystemMatrix2C}) represents componentwise between vectors.
Due to the symmetry of $\Sigma$ and the uniqueness in the sense of Remark \ref{uniqueness}, we have that $\hat{\psi}_0=c\varphi_1$ and $\psi_0=\frac{1}{c} \hat{\varphi}_1$ for some positive scalar $c$. It readily follows that
\begin{align*} \diag(\pi)\Pi^*
&=\diag(\pi)\diag(\varphi_0)^{-1}M\diag(\varphi_1)\\
&=\diag(\hat{\varphi}_0)M\diag(\varphi_1)\\
&=\diag(\hat{\psi}_0)\diag(\mu)M\diag(\varphi_1)\\
&=\diag(\hat{\psi}_0)M'\diag(\mu)\diag(\varphi_1)\\
&=\diag(\hat{\psi}_0)M'\diag(\varphi_0)^{-1}\diag(\psi_0)\diag(\varphi_1)\\
&=\diag(\varphi_1)M'\diag(\varphi_0)^{-1}\diag(\hat{\varphi}_1)\diag(\varphi_1)\\&=\Pi'\diag(\pi).\hspace*{4.2cm}\Box
\end{align*}
 \end{proof}

 Returning now to the question raised at the end of Subsection \ref{Existence}, in the case when $A$ is symmetric, we can take $\pi_A:=A\mathds{1}$. 

\section{Cooling}\label{cooling} 
In several advanced applications such as atomic force microscopy and 
microresonators, 
it is often important to 
dampen stochastic excitations (due to vibrations, molecular dynamics, etc.) affecting
 experimental apparatuses. This is often accomplished through feedback, steering the system to a desired steady state corresponding to an {\em effective temperature} which is lower than that of the environment (e.g., fluid), see for instance \cite{LMC,DE,BBP,MG,SR,Vin}. Optimal asymptotic and fast cooling of stochastic oscillators was studied in \cite{CGPcooling}. We show next how the results of Section \ref{OSSS} permit to derive corresponding results in the context of discrete space and time considered herein.

\subsection{Optimal fast cooling}\label{FC}
Consider the setting of the previous section where the transition matrix of the prior measure $M$ is time invariant. Let us introduce a {\em Boltzmann distribution}  $\pi_T$ on $\mathcal X$
\begin{equation}\label{BOL}
\pi_T(i):=Z(T)^{-1}\exp\left[\frac{-E_i}{kT}\right], \quad Z(T)=\sum_i \exp\left[\frac{-E_i}{kT}\right],
\end{equation}
corresponding to the ``energy function'' $E_x$, for $x\in\mathcal X$,
and let 
us assume that $\pi_T$ is invariant for $M=P(T)$, that is,
\[P(T)'\pi_T=\pi_T.
\]

One such class of transition matrices is given by $P(T,Q)=\left(p_{ij}(T,Q)\right)$ with

$$p_{ij}(T,Q)=\left\{\begin{array}{ll} q_{ij}\min\left(\exp\left(\frac{E_i-E_j}{kT}\right),1\right), &i\neq j,\\1-\sum_{l,l\neq i}q_{il}\min\left(\exp\left(\frac{E_i-E_l}{kT}\right),1\right), & i=j.\end{array}\right.,
$$
where $Q=\left(q_{ij}\right)$ is any symmetric transition matrix of an irreducible chain compatible with the topology of ${\bf G}$\footnote{As is well-known, this (Metropolis) chain is actually reversible with respect to the Boltzmann distribution.}. Let $T_{{\rm eff}}<T$ be a desired lower temperature. Given any prior, such as the Ruelle-Bowen measure \cite{CGPT1} or the invariant path space measure $\mathfrak P(T)$  with transition $P(T)$ and marginals $\pi_T$ and a sufficiently long time interval $[0,N]$, we can use a standard Schr\"odinger bridge (see Section \ref{GBProblems}) to steer optimally the network flow from any initial distribution $\nu_0$ to $\nu_N=\pi_{T_{{\rm eff}}}$ at time $N$. At time $N$, however, we need to change the transition mechanism to keep the chain in the steady state $\pi_{T_{{\rm eff}}}$. This is accomplished in the next subsection using the results of Section \ref{OSSS}.

\subsection{Optimal asymptotic cooling}\label{AC}
Consider once again the case where the prior transition $M=P(T)$ remains constant over time and has  the Boltzmann distribution (\ref{BOL}) as initial, invariant measure. Consider the equivalent Problems \ref{P1} and \ref{P2}. 
\begin{theorem}Assume that $P(T)$ is fully indecomposable. Then, the solution to Problem \ref{P2} is given by
\begin{equation}\label{OPTNEWTRANS}
\Pi^*=\diag(\varphi_0)^{-1}P(T)\diag(\varphi_1).
\end{equation}
where
 \begin{subequations}\label{eq:DSchroesystemMatrix3}
\begin{eqnarray}\label{eq:DSchroesystemMatrix3A}\varphi_0&=&P(T)\varphi_1,\\\label{eq:DSchroesystemMatrix3B}\hat{\varphi}_1&=&P(T)'\hat{\varphi}_0,\\
\label{eq:DSchroesystemMatrix3C}\varphi_0\circ\hat{\varphi}_0&=&\varphi_1\circ\hat{\varphi}_1=\pi_{T_{{\rm eff}}}.
\end{eqnarray}
\end{subequations}
As earlier, $\circ$ denotes componentwise multiplication of vectors. Moreover, if $P(T)$ is reversible with respect to $\pi_T$, so is $\Pi^*$ with respect to $\pi_{T_{{\rm eff}}}$.
\end{theorem}
\begin{proof}
The result follows from Theorem \ref{maincl}, Proposition \ref{PROP} and Theorem \ref{revers}. 
\end{proof}

\section{Conclusion}
The key point of this paper has been to explore analogues, in the setting of discrete space and time, for
Markovian evolutions that match specified marginals and echo well-known results in the continuous time setting
for Sch\"odinger bridges. Specifically, a key result is to show reversibility of the transition rates under the assumption that this is a property of the prior. We consider both, transitioning between marginals over a specified time window as well as the problem to secure a stationary distribution. Both cases are accomplished  by suitably adjusting the transition rates from a given prior, while keeping the distance of the new law closest to that corresponding to the prior in the relative entropy sense. Application of the results in the context regulating the flow of resources over a network is contemplated. A notion of temperature when distributions are expressed as Boltzmann distributions with respect to an energy function is brought up, and the problem to transition between distributions corresponding to different temperatures considered, in analogy to the problem of cooling in continuous time and space.


%
\appendix[Proof of Theorem \ref{maincl}]\label{A}
We form the Lagrangian for Problem \ref{problemfluid-dynamic}:
\begin{eqnarray}\nonumber&&\mathcal L(p,\pi;\lambda)=\sum_{t=0}^{N-1}\sum_{x_t}\D(\pi_{x_tx_{t+1}}(t)\|m_{x_tx_{t+1}}(t))p_t(x_t)\\&&+\sum_{t=0}^{N-1}\sum_{x_{t+1}}\lambda_t(x_{t+1})\left[p_{t+1}(x_{t+1})-\sum_{x_t} p_t(x_t) \pi_{x_{t}x_{t+1}}(t)\right]\nonumber
\end{eqnarray}
 We now use discrete integration by parts:
\begin{eqnarray}\nonumber&&\sum_{t=0}^{N-1}\sum_{x_{t+1}}\lambda_t(x_{t+1})p_{t+1}(x_{t+1})\\\nonumber&&=\sum_{t=0}^{N-1}\sum_{x_{t}}\lambda_{t-1}(x_{t})p_{t}(x_{t})\\\nonumber&&+\sum_{x_N}\lambda_{N-1}(x_N)\nu_N(x_N)-\sum_{x_0}\lambda_{-1}(x_0)\nu_0(x_0).
\end{eqnarray}
Thus, the Lagrangian reads
\begin{eqnarray}\nonumber&&\mathcal L(p,\pi;\lambda)=\sum_{t=0}^{N-1}\sum_{x_t}\D(\pi_{x_tx_{t+1}}(t)\|m_{x_tx_{t+1}}(t))p_t(x_t)\\&&+\sum_{t=0}^{N-1}\sum_{x_{t}}\left[\lambda_{t-1}(x_{t})- \sum_{x_{t+1}}\lambda_t(x_{t+1})\pi_{x_{t}x_{t+1}}(t)\right]p_t(x_t)\nonumber
\\&&+\sum_{x_N}\lambda_{N-1}(x_N)\nu_N(x_N)-\sum_{x_0}\lambda_{-1}(x_0)\nu_0(x_0).\nonumber
\end{eqnarray}
Observe that the last two terms are invariant over pairs $(p,\pi)$ satisfying (\ref{fluid-dynamic2})-(\ref{fluid-dynamic4}) and can therefore be discarded. Since we are minimizing over pairs satisfying the nonnegativity requirement and with $\pi$ satisfying (\ref{fluid-dynamic3}), we can multiply $\lambda_{t-1}(x_t)$ by $\sum_{x_{t+1}}\pi_{x_{t}x_{t+1}}(t)$ and subtract $N=\sum_{t=0}^{N-1}\sum_{x_{t}}\sum_{x_{t+1}}\pi_{x_{t}x_{t+1}}(t)p_t(x_t)$ to get
\begin{eqnarray}\nonumber&&\mathcal L(p,\pi;\lambda)=\sum_{t=0}^{N-1}\sum_{x_t}\D(\pi_{x_tx_{t+1}}(t)\|m_{x_tx_{t+1}}(t))p_t(x_t)\\&&+\sum_{t=0}^{N-1}\sum_{x_{t}}\sum_{x_{t+1}}\pi_{x_{t}x_{t+1}}(t)\left[\lambda_{t-1}(x_{t})- \lambda_t(x_{t+1})\right]p_t(x_t)\nonumber
\\&&-\sum_{t=0}^{N-1}\sum_{x_{t}}\sum_{x_{t+1}}\pi_{x_{t}x_{t+1}}(t)p_t(x_t).\nonumber
\end{eqnarray}

Next, we perform two-stage unconstrained minimization of the Lagrangian:
For fixed flow $p=\{p_t\}$ (and fixed multipliers), we optimize the strictly convex function $L(p,\cdot;\lambda)$ at each time $t$. We get the optimality conditions:
\begin{eqnarray}\nonumber
&&\left[1+\log\pi^*_{x_{t}x_{t+1}}(t)-\log m_{x_{t}x_{t+1}}(t)\right.\\\nonumber&&\left.+\lambda_{t-1}(x_t)-\lambda_{t}(x_{t+1})-1\right]p_t(x_t)=0,
\end{eqnarray}
or
\begin{equation}\label{optimalpi}
\pi^*_{x_{t}x_{t+1}}(t)=m_{x_{t}x_{t+1}}(t)\exp\left[\lambda_{t}(x_{t+1})-\lambda_{t-1}(x_t)\right].
\end{equation}
Define the positive function
\[\varphi(t+1,x_{t+1})=\exp\left[\lambda_t(x_{t+1})\right].
\]
Then (\ref{optimalpi}) becomes
\[\pi^*_{x_{t}x_{t+1}}(t)=m_{x_{t}x_{t+1}}(t)\frac{\varphi(t+1,x_{t+1})}{\varphi(t,x_{t})}.
\]
Let us impose condition  (\ref{fluid-dynamic3}) on $\pi^*$. We get that $\pi^*$ satisfies such condition if and only if $\varphi$ satisfies the reverse-time recursion
\begin{eqnarray}\nonumber
\varphi(t,x_t)&=&\sum_{x_{t+1}}m_{x_{t}x_{t+1}}(t)\varphi(t+1,x_{t+1}),\\\nonumber &&\forall x_t\in\mathcal X, \quad  t=0,1,\ldots, N-1.
\end{eqnarray}
Observe now that
\begin{eqnarray}\nonumber\mathcal L(p,\pi^*;\log\varphi)&=&\sum_{t=0}^{N-1}\sum_{x_t}\sum_{x_{t+1}}\left[-m_{x_{t}x_{t+1}}(t)\frac{\varphi(t+1,x_{t+1})}{\varphi(t,x_{t})}\right]p_t(x_t)\\\nonumber&=&-\sum_{t=0}^{N-1}\sum_{x_t}p_t(x_t)=-N,
\end{eqnarray}
which is independent of $p$. Thus, we can choose $p^*$ defined by
$$p_{t+1}^*(x_{t+1})= \sum_{x_t} p^*_t(x_t) \pi^*_{x_{t}x_{t+1}}(t),\quad p_0(x_0)=\nu(x_0),
$$
so as to satisfy constraint (\ref{fluid-dynamic2}) and the first of (\ref{fluid-dynamic4}).
Define
$$\hat{\varphi}(t,x_t):=\frac{p^*_t(x_t)}{\varphi(t,x_t)}.
$$
We get 
\begin{eqnarray}\nonumber
\hat{\varphi}(t+1,x_{t+1})&=&\frac{p^*_{t+1}(x_{t+1})}{\varphi(t+1,x_{t+1})}\\\nonumber&=&\frac{ \sum_{x_t} p^*_t(x_t) \pi^*_{x_{t}x_{t+1}}(t)}{\varphi(t+1,x_{t+1})}\\\nonumber&=&\frac{ \sum_{x_t} p^*_t(x_t) m_{x_{t}x_{t+1}}(t)\frac{\varphi(t+1,x_{t+1})}{\varphi(t,x_{t})}}{\varphi(t+1,x_{t+1})}\\&=&\sum_{x_t}  m_{x_{t}x_{t+1}}(t)\frac{p^*_t(x_t)}{\varphi(t,x_{t})}\nonumber\\\nonumber&=& \sum_{x_t}  m_{x_{t}x_{t+1}}(t)\hat{\varphi}(t,i),
\end{eqnarray}
namely $\hat{\varphi}(t,x_t)$ is {\em space-time co-harmonic}.

\ifCLASSOPTIONcompsoc
  \section*{Acknowledgments}
  Supported in part by the
NSF under grant 1807664, 1839441, 1901599, 1942523, the AFOSR under grants FA9550-17-1-0435, and by the University of Padova Research Project CPDA 140897.
\else
  \section*{Acknowledgment}
\fi

\ifCLASSOPTIONcaptionsoff
  \newpage
\fi

\end{document}